\documentclass[11pt]{article}

\usepackage[paper=letterpaper,left=1in,right=1in,top=1in,bottom=1in]{geometry}

\usepackage{latexsym, epic, eepic}
\usepackage{amsmath}
\usepackage{amssymb}
\usepackage{algorithmic,algorithm}

\usepackage{float}

\newtheorem{theorem}{Theorem}

\newtheorem{lemma}{Lemma}

\newcommand{\qed}{\hfill$\rule{2mm}{3mm}$}
\newenvironment{proof}{\par{\noindent \bf Proof:}}{\qed \par}

\newcommand{\ignore}[1]{}

\newcommand{\Int}{{\mathbb{Z}}}
\newcommand{\Utility}{{\mathcal{U}}}
\newcommand{\Inverse}{{\mathcal{V}}}
\newcommand{\Water}{{\mathcal{W}}}
\newcommand{\Patch}{{P}}
\newcommand{\Potential}{{\Phi}}

\newcommand{\NE}{{\mathcal{E}}}
\newcommand{\Neighbors}{{\mathcal{N}}}
\newcommand{\waterlevel}{\omega}

\begin{document}

\title{The Networked Common Goods Game}
\author{Jinsong Tan\\
Department of Computer and Information Science \\
University of Pennsylvania, Philadelphia, PA 19104}
\date{}
\maketitle

\thispagestyle{empty}

\begin{abstract}
We introduce a  new class of games called the \emph{networked
common goods game} (NCGG), which generalizes the well-known
\emph{common goods game} \cite{Kagel95}. We focus on a fairly
general subclass of the game where each agent's utility functions
are the same across all goods the agent is entitled to and satisfy
certain natural properties (diminishing return and smoothness). We
give a comprehensive set of technical results listed as follows.
\begin{itemize}

\item We show the optimization problem faced by a single agent can
be solved efficiently in this subclass. The discrete version of
the problem is however NP-hard but admits a \emph{fully polynomial
time approximation scheme} (FPTAS).

\item We show uniqueness results of pure strategy Nash equilibrium
of NCGG, and that the equilibrium is fully characterized by the
structure of the network and independent of the choices and
combinations of agent utility functions.

\item We show NCGG is a \emph{potential game}, and give an
implementation of best/better response Nash dynamics that lead to
fast convergence to an $\epsilon$-approximate pure strategy Nash
equilibrium.

\item Lastly, we show the \emph{price of anarchy} of NCGG can be
as large as $\Omega(n^{1-\epsilon})$ (for any $\epsilon>0$), which
means selfish behavior in NCGG can lead to extremely inefficient
social outcomes.
\end{itemize}

\end{abstract}

\section{Introduction}

A collection of members belong to various communities. Each member
belongs to one or more communities to which she can make
contributions, either monetary or in terms of service but subject
to a budget, and in turn benefits from contributions made by other
members of the communities. The extent to which a member benefits
from a community is a function of the collective contributions
made by the members of this community.

A collection of collaborators are collaborating on various
projects. Each collaborator is collaborating on one or more
projects and each project has one or more collaborators. Each
collaborator comes with certain endowment of resources, in terms
of skills, time and energy, that she can allocate across the
projects on which she is collaborating. The extent to which a
project is successful is a function of the resources collectively
allocated to it by its collaborators, and each of its collaborator
in turn derives a utility from the successfulness of the project.

A collection of friends interact with each other, and friendships
are reinforced through mutual interactions or weakened due to the
lack of them. The more time and effort mutually devoted by two
friends in their friendship, the stronger the friendship is; the
stronger the friendship is, the more each benefit from it.
However, each friend is constrained by her time and energy and has
to decide how much to devote to each of her friends.

Suppose the community members, the collaborators and the friends
(which we collectively call \emph{agents}) are all self-interested
and interested in allocating their limited resources in a way that
maximizes their own total utility derived from the communities,
projects, and mutual friendships (which we collectively call
\emph{goods}) that they have access to. Interesting computational
and economics questions abound: Can the agents efficiently find
optimal ways to allocate their resources? Viewed as a game played
by the agents over a bipartite network, how does the network
structure affect the game? In particular, does there exist a pure
strategy Nash equilibrium? Is it unique and will myopic and
selfish behaviors of the agents lead to a pure strategy Nash
equilibrium? And how costly are these myopic and selfish
behaviors?

In this paper we address these questions by first proposing a
model that naturally captures these strategic interactions, and
then giving a comprehensive set of results to the scenario where
there is only one resource to be allocated by the agents, and the
utility an agent derives from a good to which she is entitled is a
concave and smooth function of the total resource allocated to
that good. We start by giving our model that we call the
\emph{networked common goods game} (NCGG). \vspace{0.5em}

\noindent {\bf The Model.} The \emph{networked common goods game}
is played on a bipartite graph $G = (P, A, E)$, where $P = \{p_1,
p_2, ..., p_n\}$ is a set of \emph{goods} and $A = \{a_1, a_2,
..., a_m\}$ is a set of \emph{agents}. If there is an edge $(p_i,
a_j) \in E$, then agent $a_j$ is entitled to good $p_i$. There is
a single kind of divisible resource of which each agent is endowed
with one unit (we note this is not a loss of generality as our
results generalize easily to the case where different agents start
with different amounts of resource). Moreover, we can assume
Nature has endowed each common good $p_i$ with $\alpha_i$ amount
of resource that we call the \emph{ground level}; this can be
viewed as modelling $p_i$ as having access to some external
sources of contributions.

Let $\Neighbors(v)$ denote the set of neighbors of a node $v\in P
\cup A$, $x_{ij} \in [0,1]$ the amount of resource agent $a_j$
contribute to good $p_i$ and $\waterlevel_i = \alpha_i + \sum_{a_k
\in \Neighbors(p_i)} x_{ik}$ the total amount of resource
allocated to good $p_i$. Each agent $a_j$ derives certain utility
$\Utility_j(\waterlevel_i)$ from each $p_i$ of which she is a
member. We always assume $\Utility_j(0) = 0$ and for the most part
of the paper, we consider the case where $\Utility_j(\cdot)$ is
increasing, concave and differentiable. Being self-interested,
agent $a_j$ is interested in allocating her resources across the
goods to which she is entitled in a way that maximizes her total
utility $\sum_{p_i \in \Neighbors(a_j)}
\Utility_j(\waterlevel_i)$.

\vspace{0.5em}

\noindent {\bf Our Results.} We first consider the optimization
problem faced by a single agent: Given the resources already
allocated to the goods to which agent $a_j$ is entitled, find a
way to allocate resource so that $a_j$'s total utility is
maximized. We call this the \emph{common goods problem} (CGP) and
consider both continuous and discrete versions, where the agent's
resource is either infinitely divisible or atomic.

\begin{itemize}
\item We show that for the continuous version, if
$\Utility_j(\cdot)$ is assumed to be increasing, concave and
differentiable, then CGP has an analytical solution. On the other
hand, the discrete version of CGP is NP-hard but admits an FPTAS.
\end{itemize}

We then turn to investigate the existence and uniqueness of pure
strategy Nash equilibrium of NCGG\footnote{Since we concern
ourself with only pure strategy Nash equilibrium in this paper, we
use it interchangeably with Nash equilibrium.}. We consider two
concepts of uniqueness of equilibrium, among which \emph{strong
uniqueness} is the standard concept of equilibrium uniqueness
whereas \emph{weak uniqueness} is defined as follows: For any two
equilibria $\NE$ and $\NE'$ of the game and for any good $p_i \in
P$, the total amount of resource allocated to $p_i$ is the same
under both $\NE$ and $\NE'$. We have the following results.

\begin{itemize}

\item We show for any NCGG instance, a Nash equilibrium always
exists. And we show that this Nash equilibrium is weakly unique
not only in a particular NCGG instance, but across all NCGG
instances played on the same network as long as the utility
function of each agent is increasing, concave and differentiable.
And if in addition the underlying graph is a tree, the equilibrium
is strongly unique. Our results do not assume that different
agents have the same utility function; this demonstrates that Nash
equilibrium in NCGG is completely characterized by network
structure.
\end{itemize}

We also consider the convergence of Nash dynamics of the game, and
its \emph{price of anarchy}: The worst-case ratio between the
social welfare of an optimal allocation of resources and that of a
Nash equilibrium \cite{Papa01}.

\begin{itemize}

\item We show that NCGG is a \emph{potential game}, a concept
introduced in \cite{Monderer}, therefore any (better/best
response) Nash dynamics always converge to the (unique) pure
strategy Nash equilibrium. We then propose a particular
implementation of Nash dynamics that leads to fast convergence to
a state that is an additive $\epsilon$-approximation of the pure
strategy Nash equilibrium of NCGG. The convergence takes $O(Kmn)$
time, where $K = \max_{j}{ \Utility_j^{-1}(\epsilon/n)}$, which
for most reasonable choices of $\Utility_j$ is a polynomial of $n$
and $m$. (For example, for $\Utility_j(x) = x^p$ where $p\in
(0,1)$ is a constant, it is sufficient to set $K =
(n/\epsilon)^{1/p}$, which is a polynomial in $n$.)

\item We show the price of anarchy of the game is
$\Omega(n^{1-\epsilon})$ (for any $\epsilon>0$), which means
selfish behavior in this game can lead to extremely inefficient
social outcomes, for a reason that echoes the phenomenon of
\emph{tragedy of the commons} \cite{Hardin68}.

\end{itemize}

We note that NCGG introduced in this paper has the particularly
nice property that very little is assumed about agents' utility
functions. Unlike most economic models considered in the
literature where not only a particular form of utility function is
assumed about a particular agent, but very often the same utility
function is imposed across all agents, so that the model remains
mathematically tractable, our model do not assume more than the
following: 1) $\Utility_j(0) = 0$; 2) $\Utility_j(\cdot)$ has
diminishing return (increasing and convex); 3) $\Utility_j(\cdot)$
is smooth (differentiable). In particular, we do not need to
assume different agents share a common utility function for our
results to go through.

\vspace{0.5em}

\noindent {\bf Related Work.} The networked common goods game we
consider is a natural generalization of the well-known common
goods game \cite{Kagel95}. Bramoull\'{e} and Kranton considered a
different generalization of the common goods game to networks
\cite{Bramoulle07}. In their formulation a (general,
non-bipartite) network is given where each node represents an
agent $a_i$, who can exert certain amount of effort $e_i \in
[0,+\infty)$ towards certain common good and such effort incurs a
cost of $ce_i$ on the part of the agent, for some constant $c$.
$a_i$'s effort directly benefits another agent $a_j$ iff they are
directed connected in the network, and the utility of $a_i$ is
defined as $\Utility_i(e_i + \sum_{a_j \in \Neighbors(i)}e_j) -
ce_i$. Bramoull\'{e} and Kranton then analyze this model to yield
the following interesting insights: First, in every network there
is an equilibrium where some individuals contribute whereas others
free ride. Second, specialization can be socially beneficial. And
lastly, a new link in the network can reduce social welfare as it
can provide opportunities to free ride and thus reduce individual
incentives to contribute. We note both the model and the research
perspectives are very different from those considered in this
paper.

A more closely related model is that studied by Fol'gardt
\cite{Folgardt93,Folgardt95}. The author considered a resource
allocation game played on a bipartite graph that is similar to our
setting. In Fol'gardt's model, each agent has certain amount of
discrete resources, each of unit volume, that she can allocate
across the `sites' that she has access to. Each site generates
certain utility for the agent, depending on the resources jointly
allocated to it by all its adjacent agents. In Fol'gardt's
formulation, each agent is interested in maximizing the minimum
utility obtained from a single site she has access to. The
analysis of Fol'gardt's resource allocation game is limited to
very specific and small graphs \cite{Folgardt93,Folgardt95}.

A variety of other models proposed and studied in the literature
bear similarities to the networked common goods game considered
here. These include Fisher's model of economy \cite{Fisher1891},
the bipartite exchange economy \cite{Kakade04,EvenDar07}, the
fixed budget resource allocation game \cite{Feldman05,Zhang05},
the Pari-Mutuel betting as a method of aggregating subjective
probabilities \cite{Eisenberg}, and the market share game
\cite{Goemans04}. However these model all differ significantly in
the ways allocations yield utility.

\section{The Common Goods Problem}

Recall that CGP is the optimization problem faced by a single
agent: An agent has access to $n$ goods, each good $p_i$ has
already been allocated $\alpha_i \geq 0$ resources. The agent has
certain amount of resource to allocate across the $n$ goods.
Denote by $x_i$ ($i = 1,...,n$) the amount of resource the agent
allocates to goods $i$, she receives a total utility of
$\sum_{i=1}^n \Utility(\alpha_i + x_i)$. In this section, we
consider two versions of this optimization problem, where the
resource is either infinitely divisible or discrete.

\subsection{Infinitely Divisible Resource}

Without loss of generality, assume the agent has access to one
unit of resource. In the infinitely divisible case, CGP is a
convex optimization problem captured by the following convex
program.
\begin{equation}\label{ConvexProgramGeneral}
\begin{array}{rllr}
& \mbox{ maximize } &  \sum_{i=1}^n{\Utility{(\alpha_i + x_i)}} \\
& \mbox{ subject to} & \sum_{i=1}^n{x_i} = 1 \\
& & x_i \geq 0 \qquad \qquad (i = 1,2,...,n)
\end{array}
\end{equation} where the constraint $\sum_{i=1}^n{x_i} = 1
$ comes from the observation that $\Utility(\cdot)$ is an
increasing function so an optimal solution must have allocated the
entire unit of resource.

As it turns out, as long as $\Utility(\cdot)$ is increasing,
concave and differentiable, the above convex program admits
exactly the same unique solution regardless of the particular
choice of $\Utility(\cdot)$. And we note this solution coincides
with what is known in the literature as the {\it water-filling}
algorithm \cite{Boyd04}. This is summarized in the following
theorem. The proof relies on the above program being convex to
apply the well-known \emph{Karush-Kuhn-Tucker} (KKT) optimality
condition \cite{Boyd04}, and is relegated to the appendix.

\begin{theorem}\label{Thm:UniqueSolution} For any utility function $\Utility$ that is
concave and differentiable, the convex program admits a unique
analytical solution. Moreover, the solution is unique across all
choices of $\Utility(\cdot)$ as long as it is increasing, concave
and differentiable.
\end{theorem}

Therefore the unique optimal way to allocate resources across the
goods is independent of the agent's utility function as long as it
is differentiable and has diminishing return, which is a very
reasonable assumption. We note this is a particularly nice
property of the model as it frees us from imposing any particular
form of utility function, which can often be arbitrary, and the
risk of observing artifacts thus introduced. In NCGG considered
later, this property frees us from making the assumption that each
agent has the same utility function, which is standard of most
economic models whose absence would often render the underlying
model intractable.

\subsection{Discrete Resource}

In the discrete case, the agent has access to a set of atomic
resources, each of integral volume. We show in the next two
theorems that although the discrete CGP is NP-hard even in a
rather special case, the general problem always admits an FPTAS.

\begin{theorem} \label{Thm:NPhard} The discrete common goods problem is
NP-hard even when each atomic resource is of unit volume and
$\Utility(\cdot)$ is increasing.
\end{theorem}
\begin{proof} We prove the hardness result by giving a reduction
from the NP-hard \emph{unbounded knapsack problem}
\cite{Martello90}. \vspace{0.3em}

\noindent {\sc Unbounded Knapsack Problem (UKP)}

\noindent {\sc Instance:} A finite set $U = \{1,2,...,n\}$ of
items, each item $i$ has value $v_i \in \Int^+$, weight $w_i \in
\Int^+$ and unbounded supply, a positive integer $B \in \Int^+$.

\noindent {\sc Question:} Find a multi-subset $U'$ of $U$ such
that $\sum_{i\in U'}v_i$ is maximized and $\sum_{i\in U'}w_i \leq
B$.

\vspace{0.3em}

Since supply is unlimited we can assume without loss of generality
that no two items are of the same weight and no item is strictly
dominated by any other item, i.e. $w_i > w_j$ implies $v_i
> v_j$. Now create $n$ goods, $p_1,...,p_n$, where $p_i$
corresponds to item $i$ and has a ground level $(i-1)B$. Let the
agent have access to a total of $B$ atomic resource, each of unit
volume. Define the utility function $\Utility(\cdot)$ as follows:
$\Utility(\waterlevel) = \sum_{i=1}^{\mu(\waterlevel)-1}{ \lfloor
\frac{B}{w_i} \rfloor v_i } + \lfloor \frac{ \nu(\waterlevel) }{
w_{\mu(\waterlevel)} } \rfloor v_{\mu(\waterlevel)} +
\frac{\waterlevel}{(n^2 - n +2) B^2} $ where $\mu(\waterlevel) =
\lceil \waterlevel / B \rceil$ and $\nu(\waterlevel) = \waterlevel
\bmod B$.

Clearly, $\Utility(\cdot)$ is a strictly increasing function, and
thus we only concern ourselves with those CGP solutions that
allocate all $B$ atomic units of resources. One can then verify
that there is a solution of total value $K$ to the UKP instance
iff there is a solution of total utility $\sum_{j=1}^n
\sum_{i=1}^{j-1}{ \left\lfloor \frac{B}{w_i} \right\rfloor v_i } +
\frac{1}{2B} + K$ to the corresponding CGP instance. Therefore the
discrete common goods problem is NP-hard.
\end{proof}

\begin{theorem}\label{Thm:FPTAS} The discrete common goods problem always admits an FPTAS.
\end{theorem}
\begin{proof} The discrete common goods problem can
be reduced to the \emph{multiple-choice knapsack problem}

\vspace{0.3em}

\noindent {\sc Multiple-Choice Knapsack Problem (MCKP)}

\noindent {\sc Instance:} A finite set $U = \{1,2,...,k\}$ of
items, each item $i$ has value $v_i$, weight $w_i$ and belongs to
one of $n$ classes, a capacity $B > 0$.

\noindent {\sc Question:} Find a subset $U'$ of $U$ such that
$\sum_{i\in U'}v_i$ is maximized, $\sum_{i\in U'}w_i \leq B$, and
at most one item is chosen from each of the $n$ classes.

\vspace{0.3em}

The reduction goes as follows. For a general CGP instance, where
there are $B$ atomic unit-volume resources, and goods $\{p_1, ...,
p_n\}$ such that good $p_i$ has ground level $\alpha_i$, create a
MCKP instance such that there are $n$ classes $c_1, ..., c_n$.
Class $c_i$ corresponds to good $p_i$ and has $B$ items of weight
$j$ and value $\Utility(\alpha_i + j)$, for $j = 1,...,B$. The
knapsack is of total capacity $B$.

It is not hard to see that there is a solution of total utility
$K$ to the CGP instance if and only if there is a solution of
total value $K$ to the MCKP instance. Therefore, any approximation
algorithm for the latter translates into one for the former with
the same approximation guarantee. Since an FPTAS is known for MCKP
\cite{Bansal04,Kellerer}, CGP also admits an FPTAS.
\end{proof}

\section{Pure Strategy Nash Equilibrium}

We consider in this section the existence and uniqueness of Nash
equilibrium in NCGG.

\subsection{The Existence of Nash Equilibrium}\label{sec:existence}

First we show a Nash equilibrium always exists in NCGG when the
utility functions satisfy certain niceness properties.

\begin{theorem}\label{Thm:Existence} For any NCGG instance, a pure
strategy Nash equilibrium always exists as long as $\Utility_j$ is
increasing, concave and differentiable for any agent $a_j$.
\end{theorem}
\begin{proof} Let $\deg(a_i)$ be the degree of agent $a_i$ and
$D = \sum_{a_i \in A} { deg(a_i) }$. Let $s \in [0,1]^D$ be the
state vector that corresponds to how the $m$ agents have allocated
their resources, where the $( \sum_{k=1}^{i-1} deg(a_k) )$th to
the $(\sum_{k=1}^{i} deg(a_k) )$th dimension of $s$ correspond to
$a_i$'s allocation of her resource on the $deg(a_i)$ goods she is
connected to (assume an arbitrary but fixed order of the goods
$a_i$ is connected to). Define function $f : [0,1]^D \rightarrow
[0,1]^D $ such that $f(s)$ maps to the \emph{best response} state
$s'$, where $\left( s'_{ \sum_{k=1}^{i-1} deg(a_k) }, ..., s'_{
\sum_{k=1}^{i} deg(a_k)}  \right)$ corresponds to $a_i$'s best
response. Note $s'$ is unique because each agent $a_i$'s best
response is unique by Theorem \ref{Thm:UniqueSolution}, therefore
$f(s)$ is well-defined.

It is clear that $[0,1]^D$ is compact (i.e. closed and bounded)
and convex, and $f$ is continuous. Therefore, applying Brouwer's
fixed point theorem shows that $f$ has a fixed point, which
implies NCGG has a Nash equilibrium.
\end{proof}

We note on the other hand, it is easy to see that if $\Utility_j$
is allowed to be convex, then a pure strategy Nash equilibrium may
not exist in NCGG.

\subsection{The Uniqueness of Nash Equilibrium}\label{sec:uniqueness}

We next establish uniqueness results of Nash equilibrium of NCGG
in the next two theorems. Apparently, NCGG played on a general
graph does not have a unique Nash equilibrium in the standard
sense: Consider for example the $2\times 2$ complete bipartite
graph where $P = \{p_1, p_2\}$ and $A = \{ a_1, a_2 \}$, for any
$0 \leq \delta \leq 1$, $a_1$ (resp. $a_2$) allocating $\delta$
(resp. $1-\delta$) resource on $p_1$ and $1-\delta$ (resp.
$\delta$) resource on $p_2$ constitutes a pure strategy Nash
equilibrium and therefore there are uncountably infinite many of
them. However, all these equilibria can still be considered as
equivalent to each other in the sense that they all allocate
exactly the same amount of resource to each good. And the reader
is encouraged to verify as an exercise that any Nash equilibrium
in the above NCGG instance belongs to this equivalence class.
Therefore, the Nash equilibrium is still unique, albeit in a
weaker sense.

To capture this, we thus consider two concepts of uniqueness of
equilibrium: We say an NCGG instance has a \emph{weakly unique}
equilibrium if all its equilibria allocate exactly the same amount
of resource on each good $p_i$. And if an NCGG instance has an
equilibrium that is unique in the standard sense, we call it
\emph{strongly unique}. We note the concept of weak uniqueness is
a useful one as it implies the uniqueness of each agent's utility
in equilibrium, which is really what we ultimately care about.

We show two uniqueness results in this section. The first one
establishes that NCGG has a strongly unique Nash equilibrium if
the underlying graph is a tree. The second one indicates that it
is not a coincidence that the example shown above has a weakly
unique equilibrium --- in fact, we show \emph{any} NCGG instance
has a weakly unique Nash equilibrium. Furthermore, our results
indicate that the equilibrium is a function of the structure of
the underlying graph only, and independent of the particular forms
and combinations of agents' utility functions, as long as these
functions are increasing, concave and differentiable.

\begin{theorem}\label{Thm:WeaklyUnique} The Nash equilibrium
of NCGG is \emph{weakly unique} across all networked common goods
games played on a given bipartite graph $G = (P,A,E)$, as long as
$\Utility_j$ is increasing, convex and differentiable for any
agent $a_j$.
\end{theorem}
\begin{proof} Suppose otherwise that there are two equilibria $\NE$
and $\NE'$ that have different amount of resource $\waterlevel_i$
and $\waterlevel_i'$ allocated to some good $p_i$ (throughout the
rest of the paper whenever it is clear from the context, for any
good $p_x$ we denote by $\waterlevel_x$ and $\waterlevel_x'$ the
amount of resource allocated to $p_x$ in $\NE$ and $\NE'$,
respectively). Without loss of generality assume $\waterlevel_i' <
\waterlevel_i$. Then there must exists some agent $a_j \in
\Neighbors(p_i)$ who is allocating less resource on $p_i$ in
$\NE'$ than in $\NE$, and as a result, $a_j$ must be allocating
more resource on some good $p_k \in \Neighbors(a_j) \backslash
\{p_i\}$ in $\NE'$ because in equilibrium each agent allocates all
of its resources. The fact that $a_j$ is allocating nonzero
resource on $p_i$ in $\NE$ implies $\waterlevel_i \leq
\waterlevel_k$, and for the same reason $\waterlevel_k' \leq
\waterlevel_i'$. Therefore we have $\waterlevel_k - \waterlevel_k'
\geq \waterlevel_i - \waterlevel_i' > 0$.

Now consider the following process: Starting from set $S_0 = \{
p_i \}$, add goods to $S_0$ that share an agent with $p_i$ and
whose total resource have decreased by at least $\waterlevel_i -
\waterlevel_i'$ in $\NE'$; let the new set be $S_1$. Then grow the
set further by adding goods that share an agent with some good in
$S_1$ and whose total resource are reduced by at least
$\waterlevel_i - \waterlevel_i'$ in $\NE'$. Continue this process
until no more goods can be added and let the resulting set be $S$.
By construction every good in $S$ has its total resource decreased
by at least $\waterlevel_i - \waterlevel_i'$ in $\NE'$ than in
$\NE$; in fact, it can be shown that the decrease is exactly
$\waterlevel_i - \waterlevel_i'$ for each good in $S$.

If $S = P$, then we have a contradiction immediately because if
each good in $P$ has its total resource decreased by a positive
amount in $\NE'$ then it implies the agents collectively have a
positive amount of resources not allocated, contradicting the fact
that $\NE'$ is a Nash equilibrium.

We now claim that indeed $S = P$. Suppose otherwise $P = S \cup T$
and $T \neq \emptyset$. Then, $\Neighbors(S)$, the neighboring
agents of $S$ are collectively spending less resources on $S$ in
$\NE'$ than in $\NE$, which implies there exists an agent $a \in
\Neighbors(S)$ who is allocating more resources to a good $p_t\in
T$ in $\NE'$ than in $\NE$ and less resources to a good $p_s\in S$
in $\NE'$ than in $\NE$. By an argument similar to one given
above, we have $\waterlevel_s \leq \waterlevel_t$ and
$\waterlevel_t' \leq \waterlevel_s'$, and thus $\waterlevel_t -
\waterlevel_t' \geq \waterlevel_s - \waterlevel_s' \geq
\waterlevel_i - \waterlevel_i'$. This implies that $p_t$ should be
in $S$ rather than $T$; so we must have $T = \emptyset$ or $S =
P$.

Therefore $\NE$ and $\NE'$ must be equivalent in the sense that
for any good $p_i \in P$, $\waterlevel_i = \waterlevel_i'$; this
allows us to conclude that the Nash equilibrium of NCGG on any
graph is \emph{weakly unique}.
\end{proof}

Next, we move to establish the strong uniqueness result on trees.
We need the following lemma before we proceed to the main theorem
of the section.

\begin{lemma}\label{Lem:MonotoneNE} For any instance
of NCGG on a tree $G = (P, A, E)$, let $\NE$ be a Nash equilibrium
of this game, $\alpha_i$ the ground level of $p_i \in P$ and
$\waterlevel_i$ the total resource allocated on $p_i$ in $\NE$.
For any other instance of NCGG where everything is the same except
that $\alpha_i$ is increased, if $\NE'$ is an equilibrium of this
new instance and $\waterlevel_i'$ is total resource allocated to
$p_i$ in $\NE'$, then $\waterlevel_i' \geq \waterlevel_i$.
\end{lemma}
\begin{proof} Without loss of generality assume all leafs of the
tree are goods (because a leaf agent has no choice but to allocate
all her resources to the unique good she is connected to) and root
the tree at $p_i$. Suppose $\waterlevel_i' < \waterlevel_i$. Since
$\alpha_i' > \alpha_i$, it must be the case that there exists some
agent $a_j \in \Neighbors(p_i)$ who is allocating less resource on
$p_i$ in $\NE'$ than in $\NE$, this in turn implies that $a_j$ is
allocating more resource to some good $p_k \in \Neighbors
\backslash \{p_i\}$ in $\NE'$ than in $\NE$. Therefore we have
$\waterlevel_i \leq \waterlevel_k$ and $\waterlevel_k' \leq
\waterlevel_i'$ and thus $\waterlevel_k - \waterlevel_k' \geq
\waterlevel_i - \waterlevel_i' > 0$. If $k$ is a leaf then this is
obviously a contradiction. Otherwise, we can continue the above
reasoning recursively and eventually we will reach a contradiction
by having a leaf good whose total resource decreases in $\NE'$
whereas at the same time its unique neighboring agent is
allocating more resources to it.
\end{proof}

\begin{theorem}\label{Thm:StronglyUnique} The Nash equilibrium is
\emph{strongly unique} across all NCGG played on a given tree $G =
(P,A,E)$, as long as $\Utility_j$ is increasing, convex and
differentiable for any agent $a_j$.
\end{theorem}
\begin{proof} Again without loss of generality assume leafs
are all goods. We have the following claim. \vspace{0.3em}

\noindent{\underline{Claim}}. \emph{For any NCGG instance on a
tree $G = (A,P,E)$, if there is an equilibrium $\NE$ where total
resource allocated is the same across all goods, then $\NE$ is the
strongly unique Nash equilibrium.}

\noindent{\underline{Proof}}. Suppose $\NE$ is not strongly
unique. Let $\NE'$ be a different Nash equilibrium. By Theorem
\ref{Thm:WeaklyUnique} $\NE'$ can only be weakly different from
$\NE$. Since $\NE$ and $\NE'$ are weakly different there must
exist edge $(p_i, a_j)$ such that $a_j$ is allocating different
amount of resource in $\NE$ and $\NE'$; without loss of
generality, assume $a_j$ is allocating less resource in $\NE'$
than in $\NE$. Root the tree at $p_i$, then $a_j$ must be
allocating more resource in $\NE'$ to one of its child $p_k \in
\Neighbors(a_j) \backslash \{p_i\}$. Note given the amount of
resource allocated by $a_j$ on $p_k$, the game played at the
subtree rooted at $p_k$ can be viewed as independent of the game
played in the rest of the tree, by viewing the resource allocated
by $a_j$ on $p_k$ as part of the ground level of $p_k$. Now that
the ground level has increased, by Lemma \ref{Lem:MonotoneNE} any
equilibrium on the subtree rooted at $p_k$ must not have the total
resource allocated on $p_k$ decreased, so we have $\waterlevel_k'
\geq \waterlevel_k$. If $\waterlevel_k' > \waterlevel_k $, then
this is a contradiction to weak uniqueness. If $\waterlevel_k' =
\waterlevel_k$, then one of $p_k$'s child must be allocating less
resource to $p_k$ in $\NE'$ than in $\NE$ and we can repeat the
above reasoning recursively. Continue this process until we either
reach the conclusion that $\NE$ and $\NE'$ are strongly different,
which is a contradiction, or reach a leaf good whose allocated
resource in $\NE'$ is the same as that in $\NE$ even when his
unique neighboring agent is allocating more resource to it in
$\NE'$, which is again a contradiction. \vspace{0.3em}

\noindent \underline{Resume Proof of Theorem}. We prove this
theorem by giving an induction on the size of the tree $N = |A| +
|P|$. First note the equilibrium is unique when $N \leq 2$ (in the
trivial case where either $E = \emptyset$, the claim is vacuously
true). Assume the theorem is true for any tree of size $N \leq K$,
consider the case $N = K+1$.

For any instance $G_{K+1}$ with $N = K+1$, let $\NE$ be a Nash
equilibrium (whose existence is implied by Theorem
\ref{Thm:Existence}). We want to show that $\NE$ is strongly
unique. Let $$E(\NE) = \{(p_i,a_j) ~|~ \mbox{$\waterlevel_i >
\waterlevel_k$ and $x_{jk} > 0$ in $\NE$} \}$$ If $E(\NE) =
\emptyset$ then it must be the case that the total resource
allocated is the same across all goods, and by the above claim
$\NE$ is thus strongly unique and we are through. Otherwise,
partition $G$ into subtrees by removing $E(\NE)$ from $E$. Note
the size of each subtree thus resulted is at most $N$, so by
induction they each has a strongly unique equilibrium; this
implies that if we can prove $E(\NE') = E(\NE)$ for any
equilibrium $\NE'$, then $\NE' = \NE$ and we are again through. To
this end, suppose $G_{K+1}$ has a weakly different equilibrium
$\NE'$ such that $(p_i,a_j) \in E(\NE)$ and $(p_i,a_j) \notin
E(\NE')$ and consider the following two cases. \vspace{0.3em}

\noindent {\sc Case I: $a_j$ is allocating resource to $p_i$ in
$\NE'$.} Consider the game played on the subtree of $G_{K+1}$
rooted at $p_i$ and not containing $a_j$. Since $a_j$ allocates
more resource on $p_i$ in $\NE'$ than in $\NE$, by Lemma
\ref{Lem:MonotoneNE}  $\waterlevel_i' \geq \waterlevel_i$. On the
other hand, $a_j$ must be allocating less resource to some other
good $p_k$ in $\NE'$ than in $\NE$, so again by Lemma
\ref{Lem:MonotoneNE} $\waterlevel_k \geq \waterlevel_k'$. Note we
also have $\waterlevel_i > \waterlevel_k$ and thus conclude that
$\waterlevel_i' > \waterlevel_k'$; since $a_j$ allocates non-zero
resource to $p_i$ in $\NE'$, she is not acting optimally and this
gives a contradiction to the fact that $\NE'$ is an equilibrium.

\noindent {\sc Case II: $a_j$ is not allocating resource to $p_i$
in $\NE'$.} Since the subtree rooted at $p_i$ and not containing
$a_j$ is of size at most $N-1$, by induction we have
$\waterlevel_i = \waterlevel_i'$. Since $a_j$ is allocating the
same total amount of resource to $\Neighbors(a_j) \backslash p_i$,
there exists $p_k$ on which $a_j$ is allocating nonzero resource
in $\NE$ and not allocating strictly more resource in $\NE'$ than
in $\NE$; by Lemma \ref{Lem:MonotoneNE} this implies
$\waterlevel_k' \leq \waterlevel_k$. Note we also have
$\waterlevel_i > \waterlevel_k$ because $(p_i,a_j) \in E(\NE)$,
and thus we have $\waterlevel_i' > \waterlevel_k'$. Consider the
following two cases: {\sc Case 1)} If $a_j$ allocates nonzero
resource to $p_k$ in $\NE'$ then $\waterlevel_i' = \waterlevel_k'$
because $(p_i,a_j) \notin E(\NE')$; but this is a contradiction.
{\sc Case 2)} If $a_j$ allocates zero resource to $p_k$ then there
exists good $p_l \in \Neighbors(a_j) \backslash \{p_i,p_k\}$ on
which $a_j$ is allocating strictly more resource in $\NE'$ than in
$\NE$. The fact that $(p_i,a_j) \notin E(\NE')$ implies
$\waterlevel_i' = \waterlevel_l'$, so we have $\waterlevel_l' >
\waterlevel_k'$; but this is a contradiction to the fact that
$\NE'$ is an equilibrium. \vspace{0.3em}

Now we conclude that $E(\NE) = E(\NE')$ and this completes the
proof.
\end{proof}

\section{Nash Dynamics}

Pick any utility function that is increasing, concave and
differentiable, say $\Utility(x) = \sqrt{x}$, and define potential
function $\Psi(\waterlevel_1,...,\waterlevel_n) = \sum_{i=1}^n
\sqrt{\waterlevel_i}$. It is clear that for any agent $a_j$,
whenever $a_j$ updates her allocation such that increases her
total utility, the potential increases as well. This proves the
following theorem.

\begin{theorem}\label{thm:potentialGame} NCGG is a potential
game.
\end{theorem}

Therefore, better/best response Nash dynamics always converge.
However it is not clear how fast the convergence is as the
increment in $a_j$'s total utility can be either larger or smaller
than the increment of the potential, depending both on
$\Utility_j(\cdot)$ and the amount of resources already allocated
to $a_j$'s neighboring goods. In the rest of the section, we
present a particular Nash dynamics where we can show fast
convergence to an $\epsilon$-approximate Nash equilibrium. We only
give details for the best response Nash dynamics (Algorithm
\ref{NashDynamics}), and it is easy to see the same convergence
result holds for the corresponding better response Nash dynamics
as well. To this end we consider $K$-discretized version of the
game, where each agent has access to a total of $K$ identical
\emph{atomic} resources, each of volume $1/K$. We start by giving
the following two lemmas.

\begin{lemma}\label{discretediff} A solution to the $K$-discretized
CGP is optimal iff the following two conditions are satisfied: 1)
the agent has allocated all of its $K$ atomic units of resource;
2) for any two goods $p_i, p_j \in \Patch$, $\waterlevel_i -
\waterlevel_j
> 1/K$ (where $\waterlevel_i = \alpha_i + x_i$ and $\waterlevel_j
= \alpha_j + x_j$) implies $x_i = 0$.
\end{lemma}

\begin{proof} First we prove the `only if' direction. It is
obvious that an optimal solution must have allocated all of its
$K$ atomic units of resource because the utility function is
increasing, so we focus on the proof of the second condition.
Suppose otherwise we have $p_i, p_j \in \Patch$ with
$\waterlevel_i - \waterlevel_j > 1/k$, where $\waterlevel_i =
\alpha_i + x_i$, $\waterlevel_j = \alpha_j + x_j$ and $x_i > 0$.
Construct another solution by moving one atomic unit of resource
from good $p_i$ to $p_j$ gives a new solution of total utility
strictly higher because the utility function is increasing and
concave. Therefore we have a contradiction.

Next we prove the `if' direction of the lemma. Suppose the
solution $x$ is not optimal. Let $\waterlevel_k$ and
$\waterlevel'_k$ (where $p_k \in \Patch$) denote the total
resource induced by this `suboptimal' solution and a true optimal
solution $x'$, respectively. Since an optimal solution must have
allocated all of its $K$ units of atomic resource among the goods,
it must be true that there exist $p_i, p_j \in \Patch$ such that
$\waterlevel_i - \waterlevel_i' \geq 1/K$ and $\waterlevel_j' -
\waterlevel_j \geq 1/K$, and if both inequality holds in equality,
then $\waterlevel_i' \neq \waterlevel_j$ (because otherwise $x$
and $x'$ are essentially the same, which means $x$ is already
optimal). Note $\waterlevel_j' - \waterlevel_j \geq 1/K$ implies
that good $j'$ has resource allocated to it in the optimal
solution (i.e. $x_j' \geq 1/K$), so by the `only if' part of proof
above, we must have $\waterlevel_i' \geq \waterlevel_j' - 1/K$.
Now we show that $\waterlevel_i - \waterlevel_j > 1/K$ by
considering the following two cases:

{\sc Case I: $(\waterlevel_i - \waterlevel_i') + (\waterlevel_j' -
\waterlevel_j) > 2/K$}. In this case, it is easily checked that
$\waterlevel_i - \waterlevel_j > 1/K$.

{\sc Case II: $\waterlevel_i - \waterlevel_i' = 1/K$ and
$\waterlevel_j' - \waterlevel_j = 1/K$}. As discussed above, we
must not have $\waterlevel_i' = \waterlevel_j$. In fact, we must
have $\waterlevel_i' > \waterlevel_j$ because otherwise we will
have $\waterlevel_j' = \waterlevel_j + 1/K > \waterlevel_i' +
1/K$, which is a contradiction to optimality because $x_j' \geq
1/K$. Therefore, again we have reached the conclusion that
$\waterlevel_i - \waterlevel_j > 1/K$.

Now note $\waterlevel_i -\waterlevel_i' \geq 1/K$ implies $x_i
\geq 1/K$, but this is a contradiction to $\waterlevel_i -
\waterlevel_j > 1/K$, which by assumption implies $x_i = 0$.
Therefore, $x$ must itself be an optimal solution.
\end{proof}

\begin{lemma}\label{approxlemma} For any $\epsilon > 0$, an
optimal solution to the $K$-discretized common goods problem,
where $K = 1/\Utility^{-1}(\epsilon/n)$, is an
$\epsilon$-approximation to the optimal solution in the continuous
common goods problem.
\end{lemma}
\begin{proof} Denote by $OPT$ and $OPT_K$ the optimal utility
    attained by an optimal solution in the continuous version and the
    $K$-discretized version, respectively; denote by $\Water^*$ and
    $\Water^*_K$ the set of goods to which non-zero resource is
    allocated in the two optimal solutions, respectively. By Lemma
    \ref{discretediff}, any two goods in $\Water^*_K$ must have
    their total resources allocated differ by at most $1/K$, i.e.
    $\waterlevel_{max} - \waterlevel_{min} \leq 1/K$, where
    $\waterlevel_{min} = \min\{\waterlevel_i ~|~ p_i \in \Water^*_K\}$
    and $\waterlevel_{max} = \max\{\waterlevel_i ~|~ p_i \in
    \Water^*_K\}$. Since the agent has access to $n$ goods, it must
    be the case that $\waterlevel_{min} \geq 1/n - 1/K$ because
    otherwise $\waterlevel_{max} < 1/n$. Now consider the set $\Water
    = \Water^*_K \cup \{p_i \notin \Water^*_K ~|~ \alpha_i \leq
    \waterlevel_{max} \}$ of goods whose total resource is at most
    $\waterlevel_{max}$, it is clear that: 1) $\Water^*$, the optimal
    solution to the continuous version of the problem, forms a subset
    of $\Water$; 2) $\max\{\waterlevel_i ~|~ i \in \Water^*\} \leq
    \waterlevel_{max}$.

    Now suppose we have access to an additional of $|\Water|$ atomic
    units of resource, each of volume $1/K$, construct a new
    allocation by doing the following: Start with an allocation same
    as $\Water_K^*$, then assign one atomic unit of resource to each
    good in $\Water \supseteq \Water_K^*$. It is clear from the above
    discussion that for any good $p_i \in \Water$, its total
    resource under the new allocation is at least that of the total resource allocated under
    $\Water^*$, which means the utility that we obtain under the new
    allocation, $OPT''$, is at least $OPT$. Therefore $OPT - OPT_K \leq OPT'' - OPT_K
    \leq n \Utility(1/K)$; so to upper bound $OPT - OPT_K$ by $\epsilon$, it is
    sufficient to set $K = 1/\Utility^{-1}(\epsilon/n)$.
    \end{proof}

Note for most reasonable choices of $\Utility$ (e.g. $\Utility(x)
= x^p$ where $p\in (0,1)$), $K$ is polynomial in $n$. We have the
following theorem.

\begin{algorithm}[t]
\begin{algorithmic}[1]

\STATE \texttt{// INPUT: $G$, $\alpha \succeq 0$, $\epsilon
>0$, and schedule $\sigma$}

\STATE \texttt{// OUTPUT: An $\epsilon$-approximate Nash
Equilibrium}

\STATE Start by setting $K = \max_{j\in
[m]}{\Utility_j^{-1}(\epsilon/n))}$

\STATE \texttt{// Set an arbitrary initial state $s = (s_1, s_2,
..., s_m)$}

\FOR {$j = 1$ to $m$}

\STATE $a_j$ discretizes his one unit of resource into $2K$ atomic
units, each of volume $1/2K$; arbitrarily assigns them to her
adjacent goods, resulting in $s_j$

\ENDFOR

\STATE \texttt{// Sort in non-increasing order of total resource
allocated}

\STATE Arrange goods in the order $p_{\pi(1)}, p_{\pi(2)}, ...,
p_{\pi(n)}$ s.t. $\waterlevel_{\pi(i)} \geq \waterlevel_\pi(j)$ if
$1 \leq i < j \leq n$

\STATE \texttt{// Best response Nash Dynamics}

\FOR {$t = 1$ to $T$}

\STATE Let $a_{\sigma(t)}$ be the agent \emph{active} in round
$t$;

\WHILE {$\exists~ 1 \leq i < j \leq n$ s.t. $\waterlevel_{\pi(i)}
- \waterlevel_{\pi(j)} \geq 1/K$ and $x_{\sigma(t) \pi(i)}
> 0$}

\STATE $x_{\sigma(t) \pi(i)} = x_{\sigma(t) \pi(i)} - 1/2K$;
$x_{\sigma(t) \pi(j)} = x_{\sigma(t) \pi(j)} + 1/2K$

\STATE If necessary, re-define $\pi$ to maintain total resource
allocated in non-increasing order.

\ENDWHILE

\ENDFOR

\end{algorithmic}
\caption{$K$-discretized Best Response Nash Dynamics}
\label{NashDynamics}
\end{algorithm}

\begin{theorem} For any $\epsilon > 0$, Algorithm
\ref{NashDynamics} converges to an $\epsilon$-approximate Nash
equilibrium in $O(Kmn)$ time, where $K = \max_{j\in
[m]}{\Utility_j^{-1}(\epsilon/n))}$, for any updating schedule
$\sigma$.\footnote{$\sigma$ is assumed to at any time only pick an
agent whose state is not already a best-response.}
\end{theorem}
\begin{proof} First note according to the characterization of
Lemma \ref{discretediff}, the response of each agent
$a_{\sigma(t)}$ in Algorithm \ref{NashDynamics} is a
$K$-discretized best response. The rest of this proof is to define
a potential function\footnote{This potential function is different
from the one given in the proof of Theorem
\ref{thm:potentialGame}; this new potential function is convenient
in upper bounding the convergence time.} whose range are positive
integers that span an interval no greater than $Kmn$, and to show
each time an agent updates his allocation with a best response,
the value of this potential function strictly decreases.

For simplicity of exposition, we write $p_i$ in place of
$p_{\pi(i)}$ in the rest of the proof. Let $p_1, p_2, ..., p_n$ be
the $n$ goods arranged in non-increasing order of total resource
allocated, that is, $\waterlevel_{1} \geq \waterlevel_{2} \geq ...
\geq \waterlevel_{n}$. Define potential function
$\Potential(\waterlevel_{1}, \waterlevel_{2}, ...,
\waterlevel_{n}) = \sum_{i=1}^n (n-i) \cdot \waterlevel_{i}$.
Apparently, $\Potential(\cdot)$ is a positive integer valued
function and the difference between the greatest and smallest
function value is upper bounded by $Kmn$. We are done if we can
show that for any node $a_{\sigma(t)}$, the computation that
$a_{\sigma(t)}$ does on line 13-17 of Algorithm \ref{NashDynamics}
results in a strict decrease in the potential.

On line 14-15 of Algorithm \ref{NashDynamics}, an atomic unit of
resource of volume $1/2K$ is moved from good $p_i$ to $p_j$. In
doing so, the goods may no longer be sorted in non-increasing
order of total resource, and in this case we restore it on line 16
of Algorithm \ref{NashDynamics}, which without loss of generality
can be thought of as moving $p_i$ to the right for some $\mu \geq
0$ positions (with $\mu$ being the minimum necessary), and moving
$p_j$ to the left in the ordering for some $\nu \geq 0$ positions
(again with $\nu$ being the minimum necessary). This results in
the new ordering of the goods: $$p_1, ..., p_{i-1}, p_{i+1}, ...,
p_{i+\mu}, p_i, ..., p_j, p_{j-\nu}, ..., p_{j-1}, p_{j+1}, ...,
p_n$$ Note $p_i$ still precedes $p_j$ (i.e. $i+\mu < j-\nu$) in
this ordering because prior to line 14-15 of Algorithm
\ref{NashDynamics}, $\waterlevel_{i} - \waterlevel_{j} \geq 1/K$,
therefore, the total resource of $p_i$ is still at least that of
$p_j$ after a $1/2K$ amount of resource has been moved from $p_i$
to $p_j$. With this observation, we can analyze the change in
potential by looking at the changes of potential on $\{p_i, ...,
p_{i+\mu}\}$ and $\{p_{j-\nu}, ..., p_j\}$ separately, and ignore
the rest of the goods, whose contribution to potential remain
unchanged. Clearly, the contribution to potential from $\{p_i,
..., p_{i+\mu}\}$ decreases, and by an amount of $
\Delta\Potential_{\downarrow} = (\waterlevel_i -
\waterlevel_{i+1}) \cdot (n-i) + (\waterlevel_{i+1}
-\waterlevel_{i+2} ) \cdot (n-i-1) ~~ + ... + (\waterlevel_{i+\mu}
- \waterlevel_i + 1/2K) \cdot (n-i-\mu) \geq (n-i-\mu)/2K $.
Similarly, the contribution to potential from $\{p_{j-\nu}, ...,
p_j\}$ increases by $ \Delta\Potential_{\uparrow} =
(\waterlevel_{j-1} - \waterlevel_j) \cdot (n-j) +
(\waterlevel_{j-2} -\waterlevel_{j-1} ) \cdot (n-j+1) ~~ + ... +
(\waterlevel_j + 1/2K - \waterlevel_{j-\nu}) \cdot (n-j+\nu) \leq
(n-j+\nu)/2K $. Since $i+\mu < j-\nu$, we have
$\Delta\Potential_{\downarrow} > \Delta\Potential_{\uparrow}$,
which means the potential decrease by at least 1. Therefore, in at
most $Kmn$ steps Algorithm \ref{NashDynamics} converges to a Nash
equilibrium in the $K$-discretized game. By Lemma
\ref{approxlemma}, this constitutes an $\epsilon$-approximate Nash
equilibrium to the original game.
\end{proof}

\section{Price of Anarchy of the Game}

We show in this section the \emph{price of anarchy} of NCGG is
unbounded, and it is for a reason that echoes the well-known
phenomenon called \emph{tragedy of the commons} \cite{Hardin68}.

\begin{theorem}\label{poa} The \emph{price of anarchy} of NCGG
is $\Omega(n^{1-\epsilon})$, for any $\epsilon>0$.
\end{theorem}
\begin{proof} Consider the bipartite graph $G = (P, A, E)$ where $P = \{ p_c,
p_1, ..., p_n \}$, $A = \{ a_1, ..., a_n \}$ and $E =\{ (p_j,
a_j), $ $(p_c, a_j) ~|~ j \in [n]\}$ so that all agents share the
`common' good $p_c$ and each agent $a_j$ has a `private' good
$p_j$ to himself. Assume each agent $a_j$ has the same utility
function $\Utility$, $\alpha_i = 0$ $\forall ~a_i \in
\{p_1,...,p_n\}$ and $\alpha_c = 1$.

It is clear that it is a Nash equilibrium for every agent $a_j$ to
allocate her entire unit of resource to her private good $p_j$.
And in this case the social welfare is $2n \cdot \Utility(1) $. On
the other hand, if every agent devotes her entire unit of resource
to the common good, then the social welfare is $n \cdot
\Utility(n+1) $. Therefore the price of anarchy of this particular
example is at least $\frac{\Utility(n+1) }{ 2 \Utility(1) } =
O(\Utility(n+1))$. Since $\Utility(\cdot)$ is concave, we can set
$\Utility(x) = x^{1-\epsilon}$; therefore the theorem follows.
\end{proof}

\newpage

\begin{appendix}

\section{Proof of Theorem \ref{Thm:UniqueSolution}}

\begin{proof} Let $\lambda_i$ ($i = 1,2,...,n$) be the \emph{Lagrange
multiplier} associated with the inequality constraint $x_i \geq 0$
and $\nu$ the \emph{Lagrange multiplier} associated with the
equality constraint $\sum_{i=1}^n {x_i} = 1$. Since the above
program is convex, the following KKT optimality conditions,
\begin{subequations}
\begin{eqnarray}
\label{eq1}  \lambda_i^*  \geq 0, ~ x_i^* \geq 0 \quad & (i\in [n]) \\
\label{eq3}  \sum_{i=1}^n x_i^*  =  1 \quad & \\
\label{eq4} \lambda_i^* x_i^*  =  0 \quad  & (i\in [n]) \\
\label{eq5} -\frac{d}{d x_i} \Utility(\alpha_i + x_i^*) -
\lambda_i^* + \nu^* = 0 \quad & (i\in [n])
\end{eqnarray}
\end{subequations}
are sufficient and necessary for $x^*$ to be the optimal solution
to the (primal) convex program (\ref{ConvexProgramGeneral}) and
$(\lambda^*, \nu^*)$ the optimal solution to the associated dual
program.

Let $\Inverse$ be the inverse function of $\frac{d}{dt}\Utility$.
Note equation (\ref{eq4}) and (\ref{eq5}) implies $(-\frac{d}{d
x_i}\Utility(\alpha_i + x_i^*) + \nu^*) x_i^* = 0$; equation
(\ref{eq1}) and (\ref{eq5}) implies $-\frac{d}{d
x_i}\Utility(\alpha_i + x_i^*) + \nu^* \geq 0$, which combing with
the fact that $\Utility(\cdot)$ is convex implies $\Inverse
(\nu^*) \leq \alpha_i + x_i^*$. If $\alpha_i < \Inverse(\nu^*)$,
then $x_i^* > 0$ and thus $-\frac{d}{d x_i}\Utility(\alpha_i +
x_i^*) + \nu^* = 0$, i.e. $x_i^* = \Inverse(\nu^*) - \alpha_i$. On
the other hand, if $\alpha_i \geq \Inverse(\nu^*)$, then we must
have $x_i^* = 0$. To see why this is true, suppose otherwise
$x_i^* > 0$; this leads to $x_i^* = \Inverse(\nu^*) - \alpha_i
\leq 0$, which is a contradiction. We summarize the optimal
solution $x^*$ as follows
$$ x_i^* =
\begin{cases}
\Inverse(\nu^*) - \alpha_i \qquad & \alpha_i < \Inverse(\nu^*) \\
0 & \alpha_i \geq \Inverse(\nu^*)
\end{cases}$$ where $\nu^*$ is a solution to $
\sum_{i=1}^n \max{\{0, \Inverse(\nu^*) - \alpha_i\}} = 1$.

It is easy to see that $ \sum_{i=1}^n \max{\{0, \Inverse(\nu^*) -
\alpha_i\}} = 1$ admits a unique solution if we treat
$\Inverse(\nu^*)$ as the variable, i.e. different utility
functions only leads to different solutions of the Lagrange
multiplier $\nu^*$ but $\Inverse(\nu^*)$ remains invariant.
Therefore the optimal solution $x^*$ is unique not only of a
particular choice of $\Utility(\cdot)$, but across all utility
functions that are increasing, concave and differentiable.
\end{proof}

\end{appendix}

\end{document}